\documentclass[conference]{IEEEtran}

\usepackage{amsmath,graphicx}
\usepackage[utf8]{inputenc}
\usepackage{amssymb}
\usepackage{yhmath}
\usepackage[english]{babel}
\usepackage{multirow}
\usepackage{multicol}
\usepackage{float}
\usepackage{amsthm}
\usepackage{cite}
\newtheorem{theorem}{Theorem}
\newtheorem{lemma}{Lemma}
\newtheorem{corollary}{Corollary}[theorem]
\newtheorem{definition}{Definition}
\newcommand{\norm}[1]{\left\lVert#1\right\rVert}
\DeclareMathOperator{\tr}{tr}

\def\BibTeX{{\rm B\kern-.05em{\sc i\kern-.025em b}\kern-.08em
    T\kern-.1667em\lower.7ex\hbox{E}\kern-.125emX}}
\begin{document}

\title{Quality Modeling Under A Relaxed Natural Scene Statistics Model}

\author{\IEEEauthorblockN{Abhinau K. Venkataramanan}
\IEEEauthorblockA{\textit{Department of Electrical and Computer Engineering} \\
\textit{The University of Texas at Austin}\\
Austin, USA \\
abhinaukumar@utexas.edu}
\and
\IEEEauthorblockN{Alan C. Bovik}
\IEEEauthorblockA{\textit{Department of Electrical and Computer Engineering} \\
\textit{The University of Texas at Austin}\\
Austin, USA \\
bovik@ece.utexas.edu}
}

\maketitle

\begin{abstract}
Information-theoretic image quality assessment (IQA) models such as Visual Information Fidelity (VIF) and Spatio-temporal Reduced Reference Entropic Differences (ST-RRED) have enjoyed great success by seamlessly integrating natural scene statistics (NSS) with information theory. The Gaussian Scale Mixture (GSM) model that governs the wavelet subband coefficients of natural images forms the foundation for these algorithms. However, the explosion of user-generated content on social media, which is typically distorted by one or more of many possible unknown impairments, has revealed the limitations of NSS-based IQA models that rely on the simple GSM model. Here, we seek to elaborate the VIF index by deriving useful properties of the Multivariate Generalized Gaussian Distribution (MGGD), and using them to study the behavior of VIF under a Generalized GSM (GGSM) model.
\end{abstract}

\section{Introduction}
\label{sec:intro}
The field of Natural Scene Statistics (NSS) modeling of images originates from the seminal work of Ruderman \cite{ref:ruderman_nss}, later formalized in the form of the Gaussian Scale Mixture (GSM) model that governs the subband coefficients of natural images \cite{ref:gsm_nss}. The GSM model forms the backbone of successful Full-Reference (FR) Image Quality Assessment (IQA) models such as Information Fidelity Criterion (IFC) \cite{ref:ifc}, Visual Information Fidelity (VIF) \cite{ref:vif}, and VMAF \cite{ref:vmaf}, as well as Reduced Reference (RR) Video Quality Assessment (VQA) models such as ST-RRED \cite{ref:strred} and SpEED-QA \cite{ref:speedqa}. The GSM model also underlies the perceptually significant divisive normalization processess used in popular No-Reference (NR) IQA models like DIIVINE \cite{ref:diivine}, BRISQUE \cite{ref:brisque}, and NIQE \cite{ref:niqe}.

However, images may undergo distortions in multiple stages (e.g. acquisition and compression) or may coincidentally occur, commingling to form aggregate distortions that are hard to model or even characterize. Distortions such as these often arise in User Generated Content (UGC) \cite{ref:live_challenge} \cite{ref:ugc_video_db}, where combined degradations are hard to capture using the GSM model. A comprehensive evaluation of NR IQA models on UGC can be found in \cite{ref:ugc_benchmark}.

Here, we provide a theoretical analysis of VIF under a Generalized GSM (GGSM) model of NSS \cite{ref:ggsm}. We achieve this by deriving novel results on the Multivariate Generalized Gaussian Distribution (MGGD). Further, we also derive results on the multivariate kurtosis, which we use to approximate the true value of VIF, assuming the images being analyzed follow a GGSM model.

\section{Visual Information Fidelity}
\label{sec:vif}
We begin by reviewing VIF, which is an information-theoretic FR IQA model. Let the coefficients obtained after a bandpass transform on the reference image, such as a steerable pyramid decomposition \cite{ref:steer}, be \(\mathbf{C_i} \in \mathbb{R}^M\), \(i = 1 \dots N\). Assuming a signal attenuation and additive noise distortion model, where \(\mathbf{D_i} = g_i \mathbf{C_i} + \mathbf{V_i}\)

are distorted bandpass coefficients, \(g_i\) are deterministic scalars, and \(\mathbf{V_i} \sim \mathcal{N}(0, \sigma_v^2\mathbf{I})\).

VIF assumes that the \(\mathbf{C_i}\) are distributed as a Gaussian Scale Mixture (GSM). That is, each \(\mathbf{C_i}\) can be expressed as \(\mathbf{C_i} = Z_i \mathbf{U_i}\), where \(Z_i\) are non-negative scalar random variables (RVs) and \(\mathbf{U_i} \sim \mathcal{N}(0, \mathbf{C_U})\), \(\mathbf{C_U} \in \mathbb{R}^{M\times M}\), independent of \(Z_i\).

In VIF, uncertainty of perception, such as neural noise is modelled as additive white Gaussian noise (AWGN), leading to observed reference and distorted bandpass coefficients \(\mathbf{E_i}\) and \(\mathbf{F_i}\), given by
\begin{equation}
    \label{eq:reference_observed}
    \mathbf{E_i} = Z_i \mathbf{U_i} + \mathbf{N_i},
\end{equation}
\begin{equation}
    \label{eq:distorted_observed}
    \mathbf{F_i} = g_i Z_i \mathbf{U_i} + \mathbf{N'_i},
\end{equation}
where \(\mathbf{N_i}, \mathbf{N'_i} \sim \mathcal{N}(0, \sigma_n^2\mathbf{I})\) are independent and identically distributed (iid). Then, assuming a \(K\)-subband decomposition, the VIF index is defined as
\begin{equation}
    \label{eq:vif_expr}
    VIF = \frac{\sum\limits_{k = 1}^{K} \sum\limits_{i = 1}^{N} I(\mathbf{C_i^k}; \mathbf{F_i^k} | Z_i^k = z_i^k)}{\sum\limits_{k = 1}^{K} \sum\limits_{i = 1}^{N} I(\mathbf{C_i^k}; \mathbf{E_i^k} | Z_i^k = z_i^k)}.
\end{equation}

\section{Multivariate Generalized Gaussian Distribution}
\label{sec:mggd}
The probability density function (pdf) \(f_\mathbf{U}\) of an \(M\)-dimensional MGGD RV \(\mathbf{U} \sim \mathcal{MGGD}(0, \alpha, \mathbf{C_U})\) is given by
\begin{equation}
    f_U(\mathbf{u}) = \frac{\alpha \Gamma(\frac{M}{2}) \det(\mathbf{C_U})^{-\frac{1}{2}}}{\pi^{\frac{M}{2}} \Gamma(\frac{M}{2\alpha}) 2^{\frac{M}{2\alpha}}} \exp\left(-\frac{1}{2} \left( \mathbf{u}^T \mathbf{C_U}^{-1} \mathbf{u}\right)^\alpha\right).
\end{equation}

If \(\mathbf{U_i} \sim \mathcal{MGGD}(0, \alpha, \mathbf{C_U})\) in \eqref{eq:reference_observed} - \eqref{eq:distorted_observed}, where \(\alpha > 0\), the resulting model is called a Generalized GSM (GGSM). When \(\alpha=1\), the GGSM and GSM models are identical.

The first result we derive in this section is an expression for the differential entropy of an MGGD, which is hitherto known only for the scalar case (\(M=1\)).
\begin{lemma}
\label{lem:mggd_entr}
The differential entropy \(h(\mathbf{U})\) of a RV \(\mathbf{U} \sim \mathcal{MGGD}(0, \alpha, \mathbf{C_U})\) is given by
\begin{equation}
    \label{eq:ggd_entr}
    h(\mathbf{U}) = \frac{M}{2\alpha} - \log \frac{\alpha \Gamma(\frac{M}{2})}{\pi^{\frac{M}{2}} \Gamma(\frac{M}{2\alpha}) 2^{\frac{M}{2\alpha}}} + \frac{1}{2}\log \det(\mathbf{C_U}).
\end{equation}
\end{lemma}
We will omit the proof of Lemma \ref{lem:mggd_entr}, since the proof of Lemma \ref{lem:mggd_fim} contains all the techniques used here.

The second result we derive is the expression for the Fisher information Matrix (FIM), under translation, of an MGGD RV. 
\begin{definition}
\label{def:fisher_inf}
If \(\mathbf{X}\) has a probability density function \(f_X\), its Fisher information matrix is defined as 
\[ \mathbf{J(X)} = E \left[ \left(\frac{\nabla f_X(\mathbf{X})}{f_X(\mathbf{X})}\right)  \left(\frac{\nabla f_X(\mathbf{X})}{f_X(\mathbf{X})}\right)^T \right], \]
where \(\nabla\) denotes the gradient operator.
\end{definition}
\begin{lemma}
\label{lem:mggd_fim}
The FIM \(\mathbf{J(U)}\) of an MGGD is finite iff \(\alpha > \frac{1}{2} - \frac{M}{4}\), and is given by
\begin{equation}
    \label{eq:ggd_fisher}
    \mathbf{J(U)} = \frac{2^{2 - \frac{1}{\alpha}}\alpha^2}{M \Gamma(\frac{M}{2\alpha})}\Gamma\left(2 + \frac{M - 2}{2\alpha}\right) \mathbf{C_U}^{-1}.
\end{equation}
\end{lemma}
This quantity is easy to derive for an MGGD when \(M=1\), but no such an expression for the multivariate case is available in the literature. This result will be used to apply Lemma \ref{lem:diff_entr_ineq} in Section \ref{sec:bounds}.
\begin{proof}
Recall that under an invertible linear transformation \(\mathbf{Y} = \mathbf{A X}\), we know that \(\mathbf{Y}\) has the pdf
\begin{equation}
    \label{eq:pdf_identity}
    f_Y(\mathbf{y}) = \frac{1}{\det{\mathbf{A}}} f_X(\mathbf{A}^{-1}\mathbf{y}).
\end{equation}
Differentiating with respect to \(\mathbf{y}\), and applying the chain rule for gradients, it follows that the FIM of \(\mathbf{Y}\) is given by
\begin{equation}
    \label{eq:fisher_inf_identity}
    \mathbf{J(Y)} = (\mathbf{A}^{-1})^T \mathbf{J(X)} \mathbf{A}^{-1}.
\end{equation}
Now, consider transformation \( \mathbf{Y} = \mathbf{C_U}^{-\frac{1}{2}}\mathbf{U} \). Then, using \eqref{eq:pdf_identity}, the \((p,q)\) th term of \(\mathbf{J(Y)}\) may be expressed
\begin{equation}
    \label{eq:fim_integral}
    [\mathbf{J(Y)}]_{pq} = \alpha^2 \int\limits_{\mathbb{R}^d} y_p y_q \norm{\mathbf{y}}_2^{4(\alpha-1)} f_Y(\mathbf{y}) d\mathbf{y}.
\end{equation}
But, when \(p \neq q\), \eqref{eq:fim_integral} is the integral of an odd function, hence,
\begin{equation}
    [\mathbf{J(Y)}]_{pq} = 0, \quad \text{if } p \neq q.
\end{equation}
When \(p = q\), the integral may be evaluated using \(M\)-spherical coordinates, which we omit the details of for brevity. Then, using \eqref{eq:fisher_inf_identity} yields the result in \eqref{eq:ggd_fisher}.

Note that the condition on \(\alpha\) is necessary and sufficient for the Gamma function in the numerator of \eqref{eq:ggd_fisher} to be defined.
\end{proof}

\section{Bounds on VIF under the GGSM Model}
\label{sec:bounds}
\subsection{Bounds on Differential Entropy}
Next, we derive bounds on VIF under the GGSM model. For simplicity, consider one subband of the bandpass decomposition, yielding random fields \(\mathbf{C}\), \(\mathbf{D}\), \(\mathbf{E}\), and \(\mathbf{F}\) defined as in Section \ref{sec:vif}. 

To derive bounds on VIF, we will use bounds on the differential entropies of sums of RVs. Let \(\mathbf{X}\) and \(\mathbf{Y}\) be two independent RVs. Then, using Shannon's Entropy Power Inequality (EPI) \cite{ref:shannon_epi}, we obtain.
\begin{equation}
    h(\mathbf{X} + \mathbf{Y}) \geq \frac{M}{2}\log\left(\exp\left(\frac{2}{M}h(\mathbf{X})\right) + \exp\left(\frac{2}{M}h(\mathbf{Y})\right) \right).
\end{equation}

To obtain an upper bound, we use the following result \cite{ref:diff_entr_bound}.
\begin{lemma}
\label{lem:diff_entr_ineq}
Let \(\mathcal{F} = \{ \mathbf{X} \in \mathbb{R}^M | E[\log(1 + \norm{\mathbf{X}}_2)] < \infty \}\), i.e., the set of RVs having finite logarithmic moments. If \(\mathbf{X} \in \mathcal{F}\) has finite \(h(\mathbf{X})\) and \(\mathbf{J(X)}\), and \(\mathbf{Z} \sim \mathcal{N}(\mu, \sigma^2 \mathbf{I})\) is independent of \(\mathbf{X}\), then
\[ h(\mathbf{X} + \mathbf{Z}) \leq h(\mathbf{X}) + \frac{M}{2}\log \left(1 + \frac{\sigma^2}{M} \tr\mathbf{J(X)}\right).\]
\end{lemma}
Since MGGD is a bounded pdf with a finite covariance, its logarithmic moment exists. Hence, \(\mathbf{U} \in \mathcal{F}\). From Section \ref{sec:mggd}, we know that \(h(\mathbf{U})\) is always finite and \(\mathbf{J(U)}\) has finite entries iff \(\alpha > \frac{1}{2} - \frac{M}{4}\). Under this mild constraint, we can apply Lemma \ref{lem:diff_entr_ineq} to the entropy terms in VIF.

\subsection{Bounds on VIF}
In \eqref{eq:vif_expr}, the mutual information term corresponding to the reference image is
\begin{equation}
    \label{eq:ref_mi}
    I(\mathbf{C_i};\mathbf{E_i} | z_i) = h(\mathbf{C_i} + \mathbf{N_i} | z_i) - h(\mathbf{N_i} | z_i).
\end{equation}

Since \(\mathbf{C_i} = Z_i U\), then \(\mathbf{C_i} |_{Z_i=z_i} \sim \mathcal{MGGD}(0, z_i^2 \mathbf{C_U}, \alpha)\). Using this relationship and the expression for the differential entropy of a Gaussian RV, we find the following lower and upper bounds on \eqref{eq:ref_mi}, respectively:
\begin{equation}
    \label{eq:ref_mi_ineq}
    I_l(\mathbf{C_i};\mathbf{E_i} | z_i) \leq I(\mathbf{C_i};\mathbf{E_i} | z_i) \leq I_u(\mathbf{C_i};\mathbf{E_i} | z_i)
\end{equation}
where
\begin{equation}
    I_l(\mathbf{C_i};\mathbf{E_i} | z_i) = \frac{M}{2}\log\left(1 + \frac{z_i^2}{2\pi e \sigma_n^2}\exp\left(\frac{2h(\mathbf{U})}{M}\right)\right),
\end{equation}
and
\begin{equation}
    I_u(\mathbf{C_i};\mathbf{E_i} | z_i) = h(\mathbf{U}) + \frac{M}{2}\log\left(\frac{z_i^2}{2 \pi e\sigma_n^2} + \frac{\tr \mathbf{J(U)}}{2 \pi e M}\right).
\end{equation}

Similarly, the mutual information term in \ref{eq:vif_expr} corresponding to the distorted image is
\begin{equation}
    \label{eq:dist_mi}
    I(\mathbf{C_i};\mathbf{F_i} | z_i) = \sum\limits_{i = 1}^{N} h(g_i \mathbf{C_i} + \mathbf{V_i} + \mathbf{N'_i} | z_i) - h(\mathbf{V_i} + \mathbf{N'_i} | z_i).
\end{equation}

As above, \(g_i\mathbf{C_i} |_{Z_i=z_i} \sim \mathcal{MGGD}(0, g_i^2z_i^2 \mathbf{C_U}, \alpha)\). Using this relationship and the expression for the differential entropy of a Gaussian RV, we derive the following lower and upper bounds on \eqref{eq:dist_mi}, respectively:
\begin{equation}
    \label{eq:dist_mi_ineq}
    I_l(\mathbf{C_i};\mathbf{F_i} | z_i) \leq I(\mathbf{C_i};\mathbf{F_i} | z_i) \leq I_u(\mathbf{C_i};\mathbf{F_i} | z_i)
\end{equation}
where
\begin{equation}
\begin{split}
    I_l(\mathbf{C_i};\mathbf{F_i} | z_i) = \frac{M}{2}&\log\left(1 + \frac{(g_iz_i)^2}{2\pi e\left(\sigma_n^2 + \sigma_v^2\right)}\exp\left(\frac{2h(\mathbf{U})}{M}\right)\right),
\end{split}
\end{equation}
and
\begin{equation}
    I_u(\mathbf{C_i};\mathbf{F_i} | z_i) = h(\mathbf{U}) + \frac{M}{2}\log\left(\frac{(g_iz_i)^2}{2 \pi e(\sigma_v^2 + \sigma_n^2)} + \frac{\tr \mathbf{J(U)}}{2 \pi e M}\right).
\end{equation}

Repeating this analysis for all \(K\) subbands, we obtain the following bound on VIF under the GGSM model.
\begin{equation}
\label{eq:vif_bound}
\begin{split}
    \frac{\sum\limits_{k = 1}^{K}\sum\limits_{i = 1}^{N}I_l(\mathbf{C_i^k};\mathbf{F_i^k} | z_i^k)}{\sum\limits_{k = 1}^{K}\sum\limits_{i = 1}^{N}I_u(\mathbf{C_i^k};\mathbf{E_i^k} | z_i^k)} &\leq \text{VIF} \leq \frac{\sum\limits_{k = 1}^{K}\sum\limits_{i = 1}^{N}I_u(\mathbf{C_i^k};\mathbf{F_i^k} | z_i^k)}{\sum\limits_{k = 1}^{K}\sum\limits_{i = 1}^{N}I_l(\mathbf{C_i^k};\mathbf{E_i^k} | z_i^k)}.
\end{split}
\end{equation}

\section{Approximating VIF In Practical Applications}
When applying this model to practical applications, like quality assessment of UGC pictures, finding bounds on the values of VIF supplies insufficient accuracy. However, obtaining an expression for VIF assuming a more general framework is difficult, since the distribution of the sum of MGGD and Gaussian RVs is not easily characterized, even in the scalar case \cite{ref:ggd_sum_neg} \cite{ref:ggd_sum}.

To overcome this limitation, \cite{ref:ggd_sum_neg} and \cite{ref:ggd_sum} approximated the distribution of the sum of independent scalar GGD RVs by a GGD RV, using a moment-matching method. Specifically, if \(U, V\) are independent scalar GGD RVs, and \(W = U + V\), then \(W\) is approximated as a GGD RV \(\tilde{W}\), such that \(W\) and \(\tilde{W}\) share the same means, variances and kurtoses. In the multivariate case, MGGDs are modelled using the mean, covariance matrix and Mardia's multivariate kurtosis \cite{ref:mardia_kurtosis} \cite{ref:mggd}. 
\begin{definition}[Mardia's Kurtosis]
Given an \(M\)-dimensional RV \(\mathbf{X}\) having mean \(\boldsymbol{\mu}\) and covariance \(\boldsymbol{\Sigma}\), Mardia's Kurtosis of \(X\) is given by
    \[\gamma_2(\mathbf{X}) = E[\left((\mathbf{X} - \boldsymbol{\mu})^T\boldsymbol{\Sigma}^{-1}(\mathbf{X} - \boldsymbol{\mu})\right)^2] - M(M+2).\]
\end{definition}

The work in \cite{ref:ggd_sum_neg} and \cite{ref:ggd_sum} is restricted to the scalar case because an expression for Mardia's kurtosis of sums of RVs is hitherto unknown. We bridge this gap by deriving the following results. In the following lemma, let \(\mathbb{S}^M_{++}\) denote the set of \(M\times M\) symmetric positive definite matrices.
\begin{lemma}
\label{lem:kurt_sum}
Let \(\mathbf{X}\) and \(\mathbf{Y}\) be independent \(M\)-dimensional RVs having zero means, without loss of generality, and covariance matrices \(\boldsymbol{\Sigma}_\mathbf{X}, \boldsymbol{\Sigma}_\mathbf{Y}\) respectively. Let \(\mathbf{Z} = \mathbf{X} + \mathbf{Y}\), having covariance \(\boldsymbol{\Sigma_Z} = \boldsymbol{\Sigma_X} + \boldsymbol{\Sigma_Y}\). Further, for any RVs \(\mathbf{A}, \mathbf{B}\) having covariances \(\boldsymbol{\Sigma_A}, \boldsymbol{\Sigma_B} \in \mathbb{S}^M_{++}\), let \(\boldsymbol{\Delta_{AB}} = \boldsymbol{\Sigma_A}^{\frac{1}{2}}\boldsymbol{\Sigma_B}^{-1}\boldsymbol{\Sigma_A}^{\frac{1}{2}}\), and
\(\rho_{\mathbf{A}\mathbf{B}} = 2\norm{\boldsymbol{\Delta_{AB}}}_F^2 + \tr\left(\mathbf{\Delta_{AB}}\right)^2.\)
Then, Mardia's Kurtosis of \(\mathbf{Z}\) sastisfies the relationship
\begin{equation}
\label{eq:kurt_sum}
\begin{split}
\gamma_2(\mathbf{Z}) = &E\left[\left(\mathbf{X}^T\boldsymbol{\Sigma_Z}^{-1}\mathbf{X}\right)^2\right] - \rho_{\mathbf{X}, \mathbf{Z}} + \\ &E\left[\left(\mathbf{Y}^T\boldsymbol{\Sigma_Z}^{-1}\mathbf{Y}\right)^2\right] - \rho_{\mathbf{Y}, \mathbf{Z}}.
\end{split}
\end{equation}
\end{lemma}

The proof of Lemma \ref{lem:kurt_sum} follows from the definitions and properties of multivariate moments and cumulants defined in \cite{ref:mult_kurt}, and evaluating expectations of Gaussian RVs. We omit the details of this proof for brevity.

\begin{definition}
A RV \(\mathbf{X}\) is said to Elliptically Distributed with mean \(\boldsymbol{\mu}\), covariance \(\boldsymbol{\Sigma}\) and a ``generating function" \(g\), i.e., \(\mathbf{X} \sim \mathcal{ED}(\boldsymbol{\mu}, \boldsymbol{\Sigma}, g)\) if the pdf of \(\mathbf{X}\) is of the form
\[f_\mathbf{X}(\mathbf{x}) = C \det(\Sigma)^{\frac{1}{2}} g((\mathbf{x} - \boldsymbol{\mu})^T \boldsymbol{\Sigma}^{-1}(\mathbf{x} - \boldsymbol{\mu})).\]
\end{definition}
\begin{theorem}
\label{th:kurt_sum_ed}
Let \(\mathbf{X} \sim \mathcal{ED}\left(0, \boldsymbol{\Sigma_X}, g_\mathbf{X}\right)\), \(\mathbf{Y} \sim \mathcal{ED}\left(0, \boldsymbol{\Sigma_Y}, g_\mathbf{Y}\right)\) be independent \(M\)-dimensional RVs, and \(\mathbf{Z}=\mathbf{X} + \mathbf{Y}\). For a generating function \(g\), let \(\lambda_{g}^{(M)} = E[V_1^4]/E[V_1^2V_2^2] - 3\), where \(V_i\) are components of the \(M\)-dimensional RV \(\mathbf{V} \sim \mathcal{ED}(0, I, g)\), and let \(^{\circ 2}\) denote elementwise squaring. Then,
\begin{align}
\label{eq:kurt_sum_ed}
    \gamma_2(\mathbf{Z}) = &\frac{\left(\rho_{\mathbf{X}\mathbf{Z}} + \lambda_{g_{\mathbf{X}}}^{(M)}\tr\left(\boldsymbol{\Delta}_{\mathbf{X}\mathbf{Z}}^{\circ 2}\right)\right) \left(\gamma_2\left(\mathbf{X}\right) + M\left(M+2\right)\right)}{M\left(M+ 2 +\lambda_{g_{\mathbf{X}}}^{(M)}\right)} + \nonumber \\
    &\frac{\left(\rho_{\mathbf{Y}\mathbf{Z}} + \lambda_{g_{\mathbf{Y}}}^{(M)}\tr\left(\boldsymbol{\Delta}_{\mathbf{Y}\mathbf{Z}}^{\circ 2}\right)\right) \left(\gamma_2\left(\mathbf{Y}\right) + M\left(M+2\right)\right)}{M\left(M+ 2 +\lambda_{g_{\mathbf{Y}}}^{(M)}\right)} - \nonumber \\
    &\rho_{\mathbf{X}\mathbf{Z}} - \rho_{\mathbf{Y}\mathbf{Z}}.
\end{align}
\end{theorem}
\begin{proof}
Let \(\mathbf{V} \sim \mathcal{ED}(0, I, g)\) be an \(M\)-dimensional RV. Observing that the distribution of \(\mathbf{V}\) is invariant to permutations of its components, it follows that \(E[Y_i^4] = E[Y_1^4]\), and \(E[Y_i^2Y_j^2] = E[Y_1^2Y_2^2]\) for all \(i\neq j\). Using this symmetry and the definition of \(\lambda_{g_\mathbf{X}}^{(M)}\), it follows that
\begin{align}
    \label{eq:spherical_dist_kurt}
    \gamma_2(\mathbf{V})&= E[\norm{\mathbf{V}}^4] - M(M+2) \nonumber\\
    &= ME[V_1^4] + M(M-1)E[V_1^2V_2^2] - M(M+2) \nonumber \\
    &= ME[V_1^2V_2^2](\lambda_{g}^{(M)} + M + 2) - M(M+2)
\end{align}

Consider the term \(E\left[\left(\mathbf{X}^T \boldsymbol{\Sigma}_\mathbf{Z}^{-1} \mathbf{X}\right)^2\right]\). Using the transformation \(\mathbf{\tilde{X}} = \boldsymbol{\Sigma_X}^{-\frac{1}{2}}\mathbf{X}\) and the definition of \(\boldsymbol{\Delta_{AB}}\), it follows that
\begin{equation}
    \label{eq:power_4_setup}
    E\left[\left(\mathbf{X}^T \boldsymbol{\Sigma}_\mathbf{Z}^{-1} \mathbf{X}\right)^2\right] = E\left[\left(\mathbf{\tilde{X}}^T\boldsymbol{\Delta_{XZ}}\mathbf{\tilde{X}}^T\right)^2\right].
\end{equation}
Note that \(\mathbf{\tilde{X}} \sim \mathcal{ED}(0, I, g_\mathbf{X})\), and since kurtosis is invariant to linear transformations, \(\gamma_2(\mathbf{\tilde{X}}) = \gamma_2(\mathbf{X})\). Expanding the polynomial in \eqref{eq:power_4_setup}, we observe that the expectation of terms having odd powers is 0. Further, by the same symmetry argument used in \eqref{eq:spherical_dist_kurt}, it follows that
\begin{align}
    \label{eq:e_xz}
    E\left[\left(\mathbf{X}^T \boldsymbol{\Sigma}_\mathbf{Z}^{-1} \mathbf{X}\right)^2\right] = &E[\tilde{X}_1^2\tilde{X}_2^2]\sum_{i\neq j} \left(2\left(\Delta_\mathbf{XZ}^{ij}\right)^2 + \Delta_\mathbf{XZ}^{ii}\Delta_{\mathbf{XZ}}^{jj}\right) \nonumber\\ &+E[\tilde{X}_1^4]\sum_{i}\left(\Delta_{\mathbf{XZ}}^{ii}\right)^2.
\end{align}
Using \eqref{eq:spherical_dist_kurt}, it follows that \eqref{eq:e_xz} evaluates to the first term in \eqref{eq:kurt_sum_ed}. Repeating the analysis for \(\mathbf{Y}\) completes the proof.
\end{proof}

For an MGGD, which is an elliptical distribution, it can be shown that \(\lambda_{g}^{(M)} = 0\). In addition, it can be shown that Mardia's Kurtosis of Gaussian RVs is zero. Using these properties, we obtain the following corollary of Theorem \ref{th:kurt_sum_ed}.
\begin{corollary}
\label{cor:kurt_sum_mggd_gauss}
Let \(\mathbf{X}\) be an MGGD RV having zero mean, shape parameter \(\alpha\) and covariance matrix \(\boldsymbol{\Sigma}\), \(\mathbf{Y} \sim \mathcal{N}(0, \sigma^2\mathbf{I})\) independent of \(\mathbf{X}\) and \(\mathbf{Z} = \mathbf{X} + \mathbf{Y}\). Then,
\[\gamma_2(\mathbf{Z}) = \gamma_2(\mathbf{X}) \frac{\rho_{\mathbf{XZ}}}{M(M+2)}.\]
\end{corollary}
Using Corollary \ref{cor:kurt_sum_mggd_gauss}, we can utilize moment-matching to approximate the distributions of \(\mathbf{E_i}\) and \(\mathbf{F_i}\) by MGGD RVs. As a result, we are able to approximate the mutual information terms in \eqref{eq:ref_mi} and \eqref{eq:dist_mi}, using the expression for differential entropy of MGGD from Lemma \ref{lem:mggd_entr}.

To illustrate this, consider \(I(\mathbf{C_i};\mathbf{E_i} | z_i)\) in \eqref{eq:ref_mi}. Note that \(h(\mathbf{N_i})\) is known, since \(\mathbf{N_i}\) is a Gaussian RV. Let \(\boldsymbol{\Sigma_U}\) be the covariance matrix of \(\mathbf{U}\). Then, given \(Z_i = z_i\), the covariance of \(\mathbf{C_i}\) is \(z_i^2 \boldsymbol{\Sigma_U}\). Using Corollary \ref{cor:kurt_sum_mggd_gauss}, the distribution of \(\mathbf{E_i} = \mathbf{C_i} + \mathbf{N_i}\) is approximated as \(\mathbf{\tilde{E_i}} \sim \mathcal{MGGD}(0, \mathbf{C_{\tilde{E_i}}}, \beta_i)\) by moment-matching, such that
\begin{equation}
    \text{Cov}(\mathbf{\tilde{E_i}}) = \text{Cov}(\mathbf{E_i}) = z_i^2\boldsymbol{\Sigma_U} + \sigma_n^2 \mathbf{I},
\end{equation}
\begin{equation}
    \gamma_2(\mathbf{\tilde{E_i}}) = \gamma_2(\mathbf{E_i}) = \gamma_2(\mathbf{U}) \frac{\rho_{\mathbf{C_i}\mathbf{E_i}}}{M(M+2)}.
\end{equation}

Knowing the covariance and kurtosis, the parameters \(\mathbf{C_{\tilde{E_i}}}\) and \(\beta_i\) may be estimated, following the standard procedure for MGGDs \cite{ref:mggd}. Using these estimates in \eqref{eq:ggd_entr}, we approximate \(h(\mathbf{C_i} + \mathbf{N_i}) \approx h(\tilde{\mathbf{E_i}})\),

using which \(I(\mathbf{C_i};\mathbf{E_i} | z_i)\) may be approximated. A similar method may be used to approximate \(I(\mathbf{C_i};\mathbf{F_i} | z_i)\) in \eqref{eq:dist_mi}. Approximate values of mutual information obtained in this manner may be used in \eqref{eq:vif_expr} to approximate VIF.

\section{Conclusion and Discussion}
In this work, we derived novel expressions for properties of MGGD RVs, using which we analyzed the VIF model under a GGSM natural image model. In addition, we proposed a method to approximate VIF using a novel result regarding Mardia's kurtosis of the sum of independent RVs. This framework extends easily to other information-theoretic quality assessment models, such as ST-RRED and SpEED-QA. To the best of our knowledge, this work is the first theoretical extension of an IQA model to a relaxed model of NSS.

The only obstacle in the way of testing this work is the lack of a multivariate GGSM modeling algorithm. The authors of \cite{ref:SEP} proposed iterative maximum likelihood and MCMC-based Bayesian methods to estimate the parameters of a scalar GGSM, under some additional assumptions on the mixing distribution \(Z_i\). An iterative modeling algorithm was proposed in \cite{ref:ggsm}, but after one iteration of that method, the divisively normalized vectors lie on the surface of an ellipsoid. As a result, no proper continuous distribution can model these vectors, making this method unstable and inaccurate in practice. Developing a sound GGSM modeling algorithm will pave the way for quality assessment models that are resilient to deviations from ideal NSS.
\bibliographystyle{IEEEbib}
\bibliography{refs}

\end{document}